\documentclass[12pt]{article}
\usepackage{amssymb}
\usepackage{eucal}
\usepackage{amsmath}
\usepackage{amsthm}
\usepackage{amssymb}
\usepackage{eucal}
\usepackage{amscd}
\usepackage{graphics}
\usepackage{graphicx}
\usepackage{amsfonts}
\usepackage{latexsym}
\usepackage[all]{xy}
\usepackage{enumerate}
\usepackage{framed}

\newcommand{\R}{{\mathbb  R}}

\numberwithin{equation}{section}
\newtheorem{thm}{\bf Theorem}[section]
\newtheorem{lem}[thm]{\bf Lemma}
\newtheorem{prop}[thm]{\bf Proposition}

\theoremstyle{remark}
\newtheorem{rem}{\bf Remark}[section]

\begin{document}

\title{The Steady States of Isotone Electric Systems}
\author{Dan Com\u anescu\\
{\small Department of Mathematics, West University of Timi\c soara}\\
{\small Bd. V. P\^ arvan, No 4, 300223 Timi\c soara, Rom\^ ania}\\
{\small E-mail address: dan.comanescu@e-uvt.ro}}
\date{}

\maketitle

\begin{abstract}
The steady states of an isotone electric system are described by an isotone function with respect to the componentwise order. When there are steady states, we highlight a dominant steady state and we study its domain of attraction for the fixed point iteration method.
\end{abstract}


\section{Introduction}

Many practical situations involve the use of electrical systems with external sources furnishing constant power to the loads. The studies of electrical systems with constant power loads (CPLs) use a wide range of mathematical techniques, see \cite{sanchez}.  A significant part of the mathematical representations of these systems is represented by the dynamical systems. In these situations the determination of the steady states and their study is of great importance.

As is presented in the paper \cite{polyak}, a linear time invariant DC system with CPLs is described as
\begin{equation}
{\bf Y}(s)=G(s){\bf U}(s)+{\bf k},
\end{equation}
where $s$ is the Laplace variable, $G(s)\in \mathcal{M}_n(\R)$, ${\bf Y}(s)=\mathcal{L}\{{\bf y}(t)\}\in \R^n$, ${\bf U}(s)=\mathcal{L}\{{\bf u}(t)\}\in \R^n$, and ${\bf k}\in \R^n$. The port variables ${\bf y,u}\in \R^n$, with the components $y_1,\dots,y_n$ and $u_1,\dots,u_n$, are connected to CPLs by
\begin{equation}
-y_i(t)u_i(t)=P_i,\,\,\,i\in \{1,\dots,n\},\,t\geq 0.
\end{equation}
A steady state $({\bf y,u})\in \R^n\times \R^n$, see \cite{polyak}, is a solution of the problem
\begin{equation}
\begin{cases}
{\bf y}=\overline{M}{\bf u}+{\bf k} \\
y_iu_i=-P_i, \,\,\,i\in \{1,\dots,n\},
\end{cases}
\end{equation}
where $\overline{M}=G(0)$ and $P_1,\dots, P_n\in \R$.
This problem has the equivalent description
\begin{equation}\label{ecuatia-fix-11} 
{\bf y}={\bf k}-\overline{M}\left({\bf P}\circ \frac{1}{\bf y}\right),
\end{equation}
where the vector $\frac{1}{\bf y}\in \R^n$ has the components $\frac{1}{y_1},\dots, \frac{1}{y_n}$, ${\bf P}$ has the components $P_1,\dots,P_n$, and ''$\circ $'' is the Hadamard product\footnote{If ${\bf y},{\bf z}\in \R^n$ have the components $y_1,\dots, y_n$ and ${z}_1,\dots, {z}_n$, then the Hadamard product (or the Schur product, \cite{chandler}) ${\bf y}\circ{\bf z}\in \R^n$ has the components $y_1{z}_1,\dots,y_n{z}_n$ (see \cite{horn}).}.
When $\overline{M}$ is an invertible matrix we can write 
\begin{equation}\label{ecuatia-matveev}
A{\bf y}+\left({\bf P}\circ \frac{1}{\bf y}\right)=A{\bf k},
\end{equation}
with $A=\overline{M}^{-1}$. This is the form of the problem studied in \cite{matveev}.  

In the papers \cite{sanchez} and \cite{polyak} it
is considered as a reasonable physical assumption the positive definiteness of the symmetric part of $\overline{M}$. In \cite{matveev} it is assumed that $A=\overline{M}^{-1}$ is a symmetric positive definite matrix and the off-diagonal elements are nonpositive (Stieltjes matrix, see \cite{micchelli}). 
In this case $\overline{M}$ is a symmetric nonnegative matrix.

In many practical situations the unknown vector ${\bf y}$ has all components with same sign.  Possibly making a change of variable we can assume that ${\bf y}$ is positive (see  \cite{matveev}).

In the paper \cite{matveev}, under the assumption that $A=\overline{M}^{-1}$ is a Stieltjes matrix, it was proven
that if positive steady states exist, then there it is a distinguished one which dominates - componentwise - all the other ones.

By using the notation $M=\overline{M}\text{diag}({\bf P})$, the equation \eqref{ecuatia-fix-11} can be written 
\begin{equation}\label{ecuatia-fix-increasing-321}
{\bf y}={\bf k}-M\frac{1}{\bf y}.
\end{equation}

In this paper we study the positive solutions of {\bf the isotone electric system} \eqref{ecuatia-fix-increasing-321} using the assumption that $M$ is a nonnegative matrix.
The positive solutions of \eqref{ecuatia-fix-increasing-321} are the fixed points of the function $T_{{\bf k},M}:\left(\R_+^*\right)^n\to \R^n$,
$T_{{\bf k},M}({\bf y})={\bf k}-M\frac{1}{\bf y}$. This is an isotone function with respect to the componentwise order\footnote{${\bf x}\leq {\bf y}\,\,\leftrightarrow \,\,x_i\leq y_i \,\,\forall i\in \{1,\dots,n\}$ (see \cite{ehrgott}).} $\leq$ on $\R^n$; for ${\bf y},\overline{\bf y}\in \left(\R_+^*\right)^n$ with ${\bf y}\leq \overline{\bf y}$ we have $T_{{\bf k},M}({\bf y})\leq T_{{\bf k},M}(\overline{\bf y})$.

In Section \ref{T-increasing-general} we highlight some properties of the fixed points of a continuous, isotone function $T:\left(\R_+^*\right)^n\to \R^n$. When $T$ is bounded from above and the set of the fixed points is nonempty, we have {\bf  the dominant fixed point of $T$} which dominates all other fixed point. Moreover, it dominates the $\omega$-limit set of $T$ for the fixed point iteration method. Also, we present a subset of the domain of attraction of the dominant fixed point of $T$ for the fixed point iteration method. 

In Section \ref{particular-case-11} we study the case when $T$ verifies 
the matrix condition \eqref{inequality-Lipschitz-123}. Such a function is a generalization of the functions $T_{{\bf k},M}$.

The Section \ref{increasing-power-563} is dedicated to the study of the fixed points of the functions $T_{{\bf k},M}$. We pay special attention to a function defined by an irreducible matrix. We apply our methods and results to the study of the steady states of a DC linear circuit with two CPLs which was studied, with other objectives and methods, in \cite{polyak} and \cite{matveev}.

At the end of the paper we present some notions and results used in the main sections.

\section{The fixed points of a continuous, isotone function}\label{T-increasing-general}

In this section we work with $T:\left(\R_+^*\right)^n\to \R^n$ a continuous, isotone function. We pay a special attention to the case when $T$ is bounded from above.

First, we present some properties of the $\omega$-limit set of $T$ for the fixed point iteration method (see Appendix \ref{fixed-point-iteration-section}). We note by $\text{Fit}_{T,{\bf y}_0}$ the sequence generated by the fixed point iteration method which start from ${\bf y}_0$. The $\omega$-limit set of $\text{Fit}_{T,{\bf y}_0}$ is $\omega_T({\bf y}_0)$. When the $\omega$-limit set of $\text{Fit}_{T,{\bf y}_0}$ has an element we also note this element by $\omega_T({\bf y}_0)$.
\begin{lem}\label{omega-limit-1}
Let be ${\bf y}_0,\overline{\bf y}_0\in D_T^{\infty}$ such that ${\bf y}_0\leq \overline{\bf y}_0$.
\begin{enumerate}[(i)]
\item If $\emph{Fit}_{T,\overline{\bf y}_0}$ is bounded from above\footnote{There exists ${\bf w}\in \left(\R_+^*\right)^n$ such that for all $r\in \mathbb{N}$ we have ${\bf y}_r\leq {\bf w}$.}, ${\bf y}\in \omega_T({\bf y}_0)$, then there exists $\overline{\bf y}\in \omega_T(\overline{\bf y}_0)$ such that ${\bf y}\leq \overline{\bf y}$.
If $\emph{Fit}_{T,\overline{\bf y}_0}$ is convergent, then $\omega_T({\bf y}_0)\subset \lfloor {\bf 0}, \omega_T(\overline{\bf y}_0)\rceil:=\{{\bf y}\in \R^n\,|\,{\bf 0}\leq {\bf y}\leq \omega_T(\overline{\bf y}_0)\}$.

\item If $\overline{\bf y}\in \omega_T(\overline{\bf y}_0)$, then there exists ${\bf y}\in \omega_T({\bf y}_0)$ such that ${\bf y}\leq \overline{\bf y}$.
 If $\emph{Fit}_{T,{\bf y}_0}$ is convergent, then $\omega_T(\overline{\bf y}_0)\subset \lfloor \omega_T({\bf y}_0), \infty\lceil:=\{{\bf y}\in \R^n\,|\,\omega_T({\bf y}_0)\leq {\bf y}\}$.
\item If $\emph{Fit}_{T,{\bf y}_0}$, $\emph{Fit}_{T,\overline{\bf y}_0}$ are convergent, then $\omega_T({\bf y}_0)\leq \omega_T(\overline{\bf y}_0)$.
\item If $\emph{Fit}_{T,{\bf y}_0}$ is convergent and $\omega_T({\bf y}_0)\in(\R_+^*)^n$, then $\omega_T({\bf y}_0)$is a fixed point of $T$.
\end{enumerate}
\end{lem}

\begin{proof}
Between the terms of $\text{Fit}_{T,{\bf y}_0}$ and $\text{Fit}_{T,\overline{\bf y}_0}$ we have ${\bf y}_r\leq \overline{\bf y}_r$, for all $r\in \mathbb{N}$ (see Lemma \ref{isotone-fixed-sequence-23}).

$(i)$ There exists the subsequence $({\bf y}_{r_q})_{q\in \mathbb{N}}$ of $\text{Fit}_{T,{\bf y}_0}$ such that ${\bf y}_{r_q}\stackrel{q\to \infty}{\longrightarrow} {\bf y}$. The subsequence $(\overline{\bf y}_{r_q})_{q\in \mathbb{N}}$ of $\text{Fit}_{T,\overline{\bf y}_0}$ is included in the compact set $\lfloor{\bf 0},{\bf w}\rceil$ where $\text{Fit}_{T,\overline{\bf y}_0}$ is bounded from above by ${\bf w}\in \left(\R_+^*\right)^n$. We can extract a subsequence of $(\overline{\bf y}_{r_q})_{q\in \mathbb{N}}$ with the limit $\overline{\bf y}$. We deduce that $\overline{\bf y}\in \omega(\overline{\bf y}_0)$ and ${\bf y}\leq \overline{\bf y}$.

$(ii)$ There exists the subsequence $(\overline{\bf y}_{r_q})_{q\in \mathbb{N}}$ of $\text{Fit}_{T,\overline{\bf y}_0}$ such that $\overline{\bf y}_{r_q}\stackrel{q\to \infty}{\longrightarrow} \overline{\bf y}$. This subsequence is bounded from above by $\overline{\bf w}$. The subsequence $({\bf y}_{k_q})_{q\in \mathbb{N}}$ of $\text{Fit}_{T,{\bf y}_0}$ is included in the compact set $\lfloor{\bf 0},\overline{\bf w}\rceil$. We can extract a subsequence of $({\bf y}_{r_q})_{q\in \mathbb{N}}$ with the limit ${\bf y}$. We deduce that ${\bf y}\in \omega_T({\bf y}_0)$ and ${\bf y}\leq \overline{\bf y}$.

For $(iii)$ we apply $(ii)$ and for $(iv)$ we use the continuity of $T$.
\end{proof}

By using the notations
\begin{equation*}
S_{T}^+:=\{{\bf y}\in (\R_+^*)^n\,|\,{\bf y}\leq T({\bf y})\},\,\,S_T^-:=\{{\bf y}\in (\R_+^*)^n\,|\,T({\bf y})\leq {\bf y}\} 
\end{equation*}
we observe that the set of the fixed points of $T$ is $\Phi_{T}= S_{T}^+\bigcap S_{T}^-$. 
In what follows we present some properties of the above sets. 

\begin{lem}\label{dominant-element-increasing}
\begin{enumerate}[(i)]
\item $S_T^+$ is invariant under $T$ (i.e. ${\bf y}\in S_T^+\Rightarrow T({\bf y})\in S_T^+$) and $S_T^+\subset D_T^{\infty}$.
\item If ${\bf y}\in S_T^+$, then $\lfloor{\bf y},\infty\lceil
$ is invariant under $T$ and $\lfloor{\bf y},\infty\lceil \subset D_T^{\infty}$.
\item If ${\bf y}\in S_T^+$ is a maximal element\footnote{If ${\bf y}\leq \overline{\bf y}$ with $\overline{\bf y}\in S_T^+$, then ${\bf y}=\overline{\bf y}$.} of $(S_T^+,\leq)$, then ${\bf y}\in \Phi_T$.
\item If ${\bf y}_0\in S_T^+$, then the fixed point iteration sequence $\emph{Fit}_{T,{\bf y}_0}$ is an isotone sequence. If this sequence is bounded from above, then it is convergent and ${\bf y}_0\leq \omega_T({\bf y}_0)\in\Phi_T$.
\item If $T$ is bounded from above\footnote{There exists ${\bf k}\in \R^n$ such that for all ${\bf y}\in \left(\R_+^*\right)^n$ we have $T({\bf y})\leq {\bf k}$.}, then $S_T^+$ is bounded from above\footnote{there exists ${\bf k}\in  \left(\R_+^*\right)^n$ such that ${\bf y}\in S_T^+\Rightarrow {\bf y}\leq {\bf k}$.}.
\end{enumerate}

\end{lem}

\begin{proof}
$(i)$ If ${\bf y}\in S_T^+$, then ${\bf y}\leq T({\bf y})$, $T({\bf y})\in \left(\R_+^*\right)^n$, and $T({\bf y})\leq T(T({\bf y}))$.

$(ii)$ For ${\bf y}\leq \overline{\bf y}$ we have ${\bf y}\leq T({\bf y})\leq T(\overline{\bf y})$.

$(iii)$ From $(i)$ we obtain $T({\bf y})\in S_T^+$. From hypothesis  ${\bf y}\leq T({\bf y})$ and ${\bf y}$ is a maximal element of $S_T^+$. We deduce that ${\bf y}=T({\bf y})$.

$(iv)$ The monotonicity of $\text{Fit}_{T,{\bf y}_0}$ can be proved by induction.  We apply Lemma \ref{omega-limit-1}.

$(v)$ If $T$ is bounded from above by ${\bf k}$ and ${\bf y}\in S_T^+$, then ${\bf y}\leq T({\bf y})\leq {\bf k}$.
\end{proof}

Analogously we obtain the following results.

\begin{lem}\label{invariant-increasing-067}
\begin{enumerate}[(i)]
\item If ${\bf y}_0\in S_T^-$, then $\emph{Fit}_{T,{\bf y}_0}$ is antitone. If ${\bf y}_0\in S_T^-\cap D_T^{\infty}$, then $\emph{Fit}_{T,{\bf y}_0}$ is convergent and $\omega_T({\bf y}_0)\in \R_+^n$.
\item If ${\bf y}_0\in S_T^-\cap D_T^{\infty}$ and $\omega_T({\bf y}_0)\in \left(\R_+^*\right)^n$, then $\omega_T({\bf y}_0)\in \Phi_T$.
\end{enumerate}
\end{lem}

In what follows we present an important result about the existence of the fixed points of $T$ and about the structure of the set $\Phi_T$.

\begin{thm}\label{x-Bix} Suppose that $S_T^+\neq \emptyset$ and it is bounded from above.

\begin{enumerate}[(i)]
\item The vector
$
{\bf y}_T^{\Box}= (\sup\limits_{{\bf y}\in S_T^+} y_1, \dots, \sup\limits_{{\bf y}\in S_T^+} y_n)^t
$
is a maximal element of $(S_T^+,\leq)$.
It is an element of  $\Phi_T$ and it dominates all the other elements of $\Phi_T$.
\item If ${\bf y}_0\in S_T^-\cap \lfloor {\bf y}_T^{\Box},\infty \lceil$, then $\emph{Fit}_{T,{\bf y}_0}$ is convergent and $\omega_T({\bf y}_0)={\bf y}_T^{\Box}$.
\end{enumerate}
\end{thm}

\begin{proof} $(i)$ We observe that ${\bf y}_T^{\Box}\in \left(\R_+^*\right)^n$ and for ${\bf y}\in S_T^+$ we have ${\bf y}\leq {\bf y}_T^{\Box}$.
First, we prove that  $\lfloor {\bf y}_T^{\Box},\infty \lceil=\bigcap_{{\bf y}\in S_T^+}\lfloor {\bf y},\infty \lceil$. We deduce that $\lfloor {\bf y}_T^{\Box},\infty \lceil\subset\bigcap_{{\bf y}\in S_T^+}\lfloor {\bf y},\infty \lceil$. If $\overline{\bf y}\in \bigcap_{{\bf y}\in S_T^+}\lfloor {\bf y},\infty \lceil$, then for all ${\bf y}\in S_T^+$ we have ${\bf y}\leq \overline{\bf y}$ and we deduce that ${\bf y}_T^{\Box}\leq \overline{\bf y}$. We obtain $\bigcap_{{\bf y}\in S_T^+}\lfloor {\bf y},\infty \lceil\subset \lfloor {\bf y}_T^{\Box},\infty \lceil$.

By using Lemma \ref{dominant-element-increasing} we obtain that $\lfloor {\bf y}_T^{\Box},\infty \lceil$ is an invariant set under $T$ and consequently, ${\bf y}_T^{\Box}\in S_T^+$. By definition of ${\bf y}_T^{\Box}$ we have that it is a maximal element of $(S_T^+,\leq)$. 
From Lemma \ref{dominant-element-increasing} we obtain that ${\bf y}_T^{\Box}$ is an element of  $\Phi_T$ and it dominates all the other element of $\Phi_T$.

$(ii)$ Because $\lfloor {\bf y}_T^{\Box},\infty \lceil$ is invariant under $T$ we deduce that $\text{Fit}_{T,{\bf y}_0}\subset \lfloor {\bf y}_T^{\Box},\infty \lceil$, ${\bf y}_0\in D_T^{\infty}$. From Lemma \ref{dominant-element-increasing} the sequence is convergent and $\omega_T({\bf y}_0)\in \lfloor {\bf y}_T^{\Box},\infty \lceil$. The unique element of $S_T^+$ in $\lfloor {\bf y}_T^{\Box},\infty \lceil$ is ${\bf y}_T^{\Box}$. 
\end{proof}

When $S_T^+\neq \emptyset$ and it is bounded from above, the fixed point ${\bf y}_T^{\Box}$ is called {\bf the dominant fixed point of $T$}. 

The set of the points ${\bf y}_0\in D_T^{\infty}$ for which $\text{Fit}_{T,{\bf y}_0}$ is convergent with the limit $\omega_T({\bf y}_0)={\bf y}\in \Phi_T$ is often called the domain of attraction of ${\bf y}$ (for the fixed point iteration method). We find an ordered segment containing the $\omega$-limit set of $T$ and we present some results about the domain of attraction of the dominant fixed point of $T$.

\begin{thm}\label{x-Bix-bounded-112}
Suppose that $T$ is bounded from above by ${\bf k}\in \left(\R^*_+\right)^n$ and $\Phi_T\neq \emptyset$.
\begin{enumerate}[(i)]
\item $S_T^+\neq \emptyset$, it is bounded from above, and ${\bf y}_T^{\Box}=(\sup\limits_{{\bf y}\in \Phi_T} y_1, \dots, \sup\limits_{{\bf y}\in \Phi_T} y_n)^t$.
\item  If ${\bf y}_0\in \lfloor {\bf k},\infty \lceil$, then $\emph{Fit}_{T,{\bf y}_0}$ is antitone and convergent and we have $\omega_T({\bf y}_0)={\bf y}_T^{\Box}$.
\item If ${\bf y}_0\in \left(\R_+^*\right)^n$, then $\omega_T({\bf y}_0)\subset \lfloor {\bf 0},{\bf y}_T^{\Box}\rceil$. The $\omega$-limit set of $T$ verifies $\Omega_T\subset \lfloor {\bf 0},{\bf y}_T^{\Box}\rceil$.
\item If ${\bf y}_0, \overline{\bf y}_0\in \left(\R_+^*\right)^n$ such that $\overline{\bf y}_0\leq {\bf y}_0$, $\emph{Fit}_{T,\overline{\bf y}_0}$ is convergent, and $\omega_T(\overline{\bf y}_0)={\bf y}_T^{\Box}$, then $\emph{Fit}_{T,{\bf y}_0}$ is convergent and $\omega_T({\bf y}_0)={\bf y}_T^{\Box}$.
\item If ${\bf y}_0\in \lfloor {\bf y}_T^{\Box},\infty \lceil$, then $\emph{Fit}_{T,{\bf y}_0}$ is a convergent sequence and its limit is $\omega_T({\bf y}_0)={\bf y}_T^{\Box}$.
\end{enumerate}
\end{thm}

\begin{proof}
$(ii)$ We observe that $\lfloor {\bf k},\infty \lceil\subset S_T^-\cap \lfloor {\bf y}_T^{\Box},\infty \lceil$ and we apply Theorem \ref{x-Bix}.

$(iii)$ There exists $\overline{\overline{\bf y}}_0\in \R^n$ such that ${\bf y}_0\leq \overline{\overline{\bf y}}_0$ and ${\bf k}\leq  \overline{\overline{\bf y}}_0$. From $(ii)$ we have that $\text{Fit}_{T,\overline{\overline{\bf y}}_0}$ is convergent and $\omega_T(\overline{\overline{\bf y}}_0)={\bf y}_T^{\Box}$. From Lemma \ref{omega-limit-1} we obtain $\omega_T({\bf y}_0)\subset \lfloor {\bf 0},{\bf y}_T^{\Box}\rceil$.

 $(iv)$ There exists $\overline{\overline{\bf y}}_0\in \R^n$ such that ${\bf y}_0\leq \overline{\overline{\bf y}}_0$ and ${\bf k}\leq  \overline{\overline{\bf y}}_0$. By induction we obtain that $\overline{\bf y}_r\leq {\bf y}_r\leq \overline{\overline{\bf y}}_r$, $r\in \mathbb{N}$. $\text{Fit}_{T,\overline{\overline{\bf y}}_0}$ is convergent and $\omega_T(\overline{\overline{\bf y}}_0)={\bf y}_T^{\Box}$. From Squeeze Theorem we deduce the announced result.  
 
$(v)$ $\text{Fit}_{T,{\bf y}_T^{\Box}}$ is convergent and its limit is $\omega_T({\bf y}_T^{\Box})={\bf y}_T^{\Box}$. We apply $(iv)$.
\end{proof}

In more restrictive conditions for $T$, we find new subsets  of domain of attraction of the dominant fixed point of $T$.

\begin{thm}\label{domain-of-attraction}
Suppose that $T$ is bounded from above, it is a concave function\footnote{$T$ is a concave function with respect to $\leq$ if $\lambda T({\bf y})+(1-\lambda)T(\overline{\bf y}) \leq T(\lambda{\bf y}+(1-\lambda)\overline{\bf y})$, when ${\bf y}, \overline{\bf y}\in \left(\R_+^*\right)^n$ and $\lambda\in (0,1)$.}, $\Phi_T$ has at least two elements, and there are no chain\footnote{A subset of a partially ordered set is a chain if it is totally ordered with respect to the induced order.} with three elements in $\Phi_T$. If \footnote{$\rfloor {\bf y}, \infty\rceil:=\{\overline{\bf y}\in \R^n\,|\,{\bf y}<\overline{\bf y}\}$}  ${\bf y}_0\in\bigcup_{{\bf y}\in \Phi_T\backslash\{{\bf y}_T^{\Box}\}}\rfloor{\bf y},\infty\lceil$, then $\emph{Fit}_{T,{\bf y}_0}$ is convergent and its limit is $\omega_T({\bf y}_0)={\bf y}_T^{\Box}$.
\end{thm}

 \begin{proof}
Let be ${\bf y}\in \Phi_T\backslash\{{\bf y}_T^{\Box}\}$ such that ${\bf y}< {\bf y}_0$. We have ${\bf y}\lneq {\bf y}_T^{\Box}$. There exists $\lambda_{\bf y}\in (0,1)$ such that $\lambda_{\bf y}{\bf y}+(1-\lambda_{\bf y}){\bf y}_T^{\Box}< {\bf y}_0$.

For $\overline{\overline{\bf y}}, \overline{\bf y}\in S_T^+$ and $\lambda\in [0,1]$ we have $\lambda\overline{\overline{\bf y}}+(1-\lambda)\overline{\bf y}\leq \lambda T(\overline{\overline{\bf y}})+(1-\lambda) T(\overline{\bf y})\leq T(\lambda\overline{\overline{\bf y}}+(1-\lambda)\overline{\bf y})$. We obtain that $S_T^+$ is a convex set and $\lambda_{\bf y}{\bf y}+(1-\lambda_{\bf y}){\bf y}_T^{\Box}\in S_T^+$.
Consequently,   
$\text{Fit}_{T,\lambda_{\bf y}{\bf y}+(1-\lambda_{\bf y}){\bf y}_T^{\Box}}$ is isotone, it is convergent and ${\bf y}\lneq \omega_T(\lambda_{\bf y}{\bf y}+(1-\lambda_{\bf y}){\bf y}_T^{\Box})\leq {\bf y}_T^{\Box}$. These vectors form a fixed point chain. We deduce that $\omega_T(\lambda_{\bf y}{\bf y}+(1-\lambda_{\bf y}){\bf y}_T^{\Box})= {\bf y}_T^{\Box}$. We apply Theorem \ref{x-Bix-bounded-112}.
 \end{proof}

We present some necessary and sufficient conditions for the existence of fixed points.

\begin{thm}
Suppose that $T$ is bounded from above by ${\bf k}\in \R^n$. The following statements are equivalent:
\begin{enumerate}[(i)]
\item $\Phi_T\neq \emptyset$.
\item $\Omega_T\cap \left(\R_+^*\right)^n\neq \emptyset$.
\item For all ${\bf y}_0$ with ${\bf k}\leq {\bf y}_0$ we have ${\bf y}_0\in D_T^{\infty}$, $\emph{Fit}_{T,{\bf y}_0}$ is convergent and $\omega_T({\bf y}_0)\in \left(\R_+^*\right)^n$.
\end{enumerate}
\end{thm}

\begin{proof}
$(i)\Rightarrow (iii)$ is obtained from Theorem \ref{x-Bix-bounded-112}. For $(iii)\Rightarrow (ii)$ we observe that $\omega_T({\bf y}_0)\in \Omega_T\cap \left(\R_+^*\right)^n$, ${\bf k}\leq {\bf y}_0$.

$(ii)\Rightarrow (i)$. Let be ${\bf y}\in \Omega_T\cap \left(\R_+^*\right)^n$. There exists ${\bf y}_0\in D_T^{\infty}$ such that ${\bf y}\in \omega_T({\bf y}_0)$. We consider $\overline{\bf y}_0$ with the properties ${\bf y}_0\leq \overline{\bf y}_0$ and ${\bf k}\leq \overline{\bf y}_0$. From Lemma \ref{isotone-fixed-sequence-23} we obtain that $\overline{\bf y}_0\in D_T^{\infty}$.
From Lemma \ref{omega-limit-1} we have $\overline{\bf y}\in \omega_T(\overline{\bf y}_0)$ with ${\bf y}\leq \overline{\bf y}$. From Lemma \ref{invariant-increasing-067} we deduce that $\text{Fit}_{T,\overline{\bf y}_0}$ is convergent and  $\omega_T(\overline{\bf y}_0)=\overline{\bf y}\in \Phi_T$.
\end{proof}

\begin{thm}\label{existence-231}
Suppose that $T$ is bounded from above and there are ${\bf y}^{\min},{\bf y}^{\max}\in \left(\R_+^*\right)^n$ so ${\bf y}^{\min}\leq {\bf y}^{\max}$, $\Phi_T\subset \lfloor{\bf y}^{\min}, {\bf y}^{\max}\rceil$, and $\lfloor {\bf y}^{\max},\infty\lceil\subset S_T^-$. The following statements are equivalent:
\begin{enumerate}[(i)]
\item $\Phi_T\neq \emptyset$.
\item For all ${\bf y}_0$ with ${\bf y}^{\max}\leq {\bf y}_0$ we have $T(\emph{Fit}_{T,{\bf y}_0})\subset \lfloor{\bf y}^{\min}, \infty\lceil$.
\end{enumerate}
\end{thm}

\begin{proof}
$(i)\Rightarrow (ii)$. Let be ${\bf y}^*\in \Phi_T$ and let be ${\bf y}_0$ with ${\bf y}^{\max}\leq {\bf y}_0$. Because ${\bf y}^*\in D_T^{\infty}$, from Lemma \ref{isotone-fixed-sequence-23}, we obtain that ${\bf y}_0\in D_T^{\infty}$. Because ${\bf y}_0\in S_T^-$, by using Lemma \ref{dominant-element-increasing}, we have that $\text{Fit}_{T,{\bf y}_0}$ is convergent and it is contained in  $\lfloor{\bf y}^{*}, \infty\lceil$. We deduce that $\omega_T({\bf y}_0)\in \Phi_T$. By hypotheses, ${\bf y}^{\min}\leq \omega_T({\bf y}_0)\leq T({\bf y}_r)$ for all $r$. 

$(ii)\Rightarrow (i)$. $\text{Fit}_{T,{\bf y}_0}$ is antitone and it is contained in  $\lfloor{\bf y}^{\min}, \infty\lceil$. We deduce that ${\bf y}_0\in D_T^{\infty}$, the sequence is convergent and $\omega_T({\bf y}_0)\in \Phi_T$.
\end{proof}

\begin{rem}
If $\text{Fit}_{T,{\bf y}_0}$ has a finite number of terms, then $\rfloor {\bf 0},{\bf y}_0\rceil\cap \Phi_T=\emptyset$. This observation can be used to delimit the set $\Phi_T$.  
\end{rem}

\section{A matrix condition}\label{particular-case-11}

In this section we suppose that the function $T:\left(\R_+^*\right)^n \to \R^n$ satisfies the equality
\begin{equation}\label{inequality-Lipschitz-123}
T(\overline{\bf y})-T({\bf y})= \boldsymbol{\boldsymbol{\mathsf{M}}}(\overline{\bf y},{\bf y})(\overline{\bf y}-{\bf y}),\, \overline{\bf y}, {\bf y}\in \left(\R_+^*\right)^n,
\end{equation}
where $\boldsymbol{\mathsf{M}}:\left(\R_+^*\right)^n\times \left(\R_+^*\right)^n\to \mathcal{M}_n(\R_+)$ is a continuous, nonnegative matrix function and $\forall \overline{\bf y}, {\bf y}\in \left(\R_+^*\right)^n$ we have $\boldsymbol{\mathsf{M}}(\overline{\bf y},{\bf y})=\boldsymbol{\mathsf{M}}({\bf y},\overline{\bf y})$. We observe that $T$ is a continuous function. 

We present some monotonicity properties of $T$ and some monotonicity properties of the fixed point iteration sequences generated by $T$.

\begin{lem}\label{*-increasing-general} 
\begin{enumerate}[(i)]
\item $T$ is a isotone function with respect to $\leq$.
\item If $\forall {\bf y}, \overline{\bf y}\in  \left(\R_+^*\right)^n$ the matrix $\boldsymbol{\mathsf{M}}(\overline{\bf y},{\bf y})$ has on each row at least a positive element, then the following results hold.
\begin{enumerate}[(a)]
\item $T$ is strongly isotone\footnote{$T$ is strongly isotone if ${\bf y}<\overline{\bf y}\Rightarrow T({\bf y})< T(\overline{\bf y})$.}.
\item If ${\bf y}_0< T({\bf y}_0)$ (respectively $T({\bf y}_0)< {\bf y}_0$), then $\emph{Fit}_{T,{\bf y}_0}$ is strongly isotone\footnote{$({\bf y}_k)_{k\in \mathcal{N}}$ is strongly isotone (respectively strongly antitone) with respect to $\leq$ if for all $k$ we have ${\bf y}_k< {\bf y}_{k+1}$ (respectively ${\bf y}_{k+1}< {\bf y}_{k}$).} (respectively strongly antitone). 
 \end{enumerate}
\end{enumerate}
\end{lem}

The spectral radius\footnote{In this paper, for a matrix ${M}\in \mathcal{M}_n(\R)$, we denote by $\rho({M})$ the spectral radius of $M$.} is a useful tool in the study of fixed points of $T$.

\begin{thm}\label{spectral-radius-strongly-12}Let be ${\bf y}_0\in D_T^{\infty}$ such that $\emph{Fit}_{T,{\bf y}_0}$is strongly monotone\footnote{A sequence is strongly monotone if it is strongly isotone or it is strongly antitone.}, convergent, and $\omega_T({\bf y}_0)\in \left(\R_+^*\right)^n$. Then,
\begin{enumerate}[(i)]
\item $\rho\left(\boldsymbol{\mathsf{M}}\left({\bf y}_0,\omega_T({\bf y}_0)\right)\right)< 1$.
\item $\rho\left(\boldsymbol{\mathsf{M}}\left(\omega_T({\bf y}_0),\omega_T({\bf y}_0)\right)\right)\leq 1$.
 \end{enumerate}
\end{thm}

\begin{proof} We consider the case when $\text{Fit}_{T,{\bf y}_0}$ is strongly isotone. 
For $r\in \mathbb{N}$ we have 
$$\omega({\bf y}_0)-{\bf y}_{r+1}=T(\omega({\bf y}_0))-T({\bf y}_k)=\boldsymbol{\mathsf{M}}(\omega({\bf y}_0),{\bf y}_r)(\omega({\bf y}_0)-{\bf y}_r).$$
Because $\omega({\bf y}_0)-{\bf y}_{r+1}< \omega({\bf y}_0)-{\bf y}_{r}$ we deduce that $\boldsymbol{\mathsf{M}}(\omega({\bf y}_0),{\bf y}_r)(\omega({\bf y}_0)-{\bf y}_r)<\omega({\bf y}_0)-{\bf y}_{r}.$ The vector $\omega({\bf y}_0)-{\bf y}_{r}$ is positive. From Corollary 8.1.29, \cite{horn}, we obtain that  
\begin{equation}\label{inequality-k-omega}
\rho\left(\boldsymbol{\mathsf{M}}\left({\bf y}_r,\omega_T({\bf y}_0)\right)\right)< 1.
\end{equation} 
When $\text{Fit}_{T,{\bf y}_0}$ is strongly antitone the proof is analogous and \eqref{inequality-k-omega} remains valid.

$(ii)$ From \eqref{inequality-k-omega} and by using the continuity of the spectral radius, see \cite{horn}, 5.6.P19, we obtain the announced result.
\end{proof}
A consequence of the above result gives us an upper bound of $\rho\left(\boldsymbol{\mathsf{M}}\left({\bf y}_T^{\Box},{\bf y}_T^{\Box}\right)\right)$.

\begin{thm}\label{inegalitate-yBox-10}
If $\Phi_T\neq \emptyset$, $T$ is bounded from above, and for all $\overline{\bf y}, {\bf y}\in \left(\R_+^*\right)^n$ the matrix $\boldsymbol{\mathsf{M}}(\overline{\bf y},{\bf y})$ has on each row at least a positive element, then  $\rho\left(\boldsymbol{\mathsf{M}}\left({\bf y}_T^{\Box},{\bf y}_T^{\Box}\right)\right)\leq 1$.
\end{thm}

\begin{proof} Let ${\bf k}\in \R^n$ be such that $T({\bf y})\leq {\bf k}$ for all ${\bf y}\in  \left(\R_+^*\right)^n$.
Let be the vector ${\bf y}_0$ such that ${\bf k}< {\bf y}_0$. We have $T({\bf y}_0)\leq {\bf k}< {\bf y}_0$.
$\text{Fit}_{T,{\bf y}_0}$ is convergent, strongly antitone (Lemma \ref{*-increasing-general}), and, from Theorem \ref{x-Bix-bounded-112}, its limit is the dominant fixed point ${\bf y}_T^{\Box}$. The inequality from the statement is the consequence of Theorem \ref{spectral-radius-strongly-12}.
\end{proof}

We present some results about the spectral radius of the matrix $\boldsymbol{\mathsf{M}}\left({\bf y},\overline{\bf y}\right)$ when ${\bf y}$ and $\overline{\bf y}$ are different fixed points of $T$.

\begin{thm}\label{q-ineq-sectral-7650-v2} Let be ${\bf y},\overline{\bf y}\in \Phi_T$ such that ${\bf y}\neq \overline{\bf y}$.
\begin{enumerate}[(i)]
\item $\rho\left(\boldsymbol{\mathsf{M}}(\overline{\bf y},{\bf y})\right)\geq 1$.
\item If ${\bf y}\in \Phi_T$ is not an isolated fixed point of $T$, then $\rho\left(\boldsymbol{\mathsf{M}}({\bf y},{\bf y})\right)\geq 1$.
\item If ${\bf y}< \overline{\bf y}$, then $\rho\left(\boldsymbol{\mathsf{M}}(\overline{\bf y},{\bf y})\right)= 1$.
\item If ${\bf y}\lneq \overline{\bf y}$ and $\boldsymbol{\mathsf{M}}(\overline{\bf y},{\bf y})$ is an irreducible matrix, then ${\bf y}< \overline{\bf y}$ and $\rho\left(\boldsymbol{\mathsf{M}}(\overline{\bf y},{\bf y})\right)= 1$.
\end{enumerate}
\end{thm}

\begin{proof}
$(i)$ From \eqref{inequality-Lipschitz-123} we have
$\overline{\bf y}-{\bf y}= \boldsymbol{\mathsf{M}}(\overline{\bf y},{\bf y})\left(\overline{\bf y}-{\bf y}\right).$ We deduce that 1 is an eigenvalue of  $\boldsymbol{\mathsf{M}}(\overline{\bf y},{\bf y})$. By definition of the spectral radius we obtain the result. 

$(ii)$ There exists the sequence $({\bf y}_r)_{r\in \mathbb{N}}$ of fixed points of $T$ such that ${\bf y}_r\neq {\bf y}$, $r\in \mathbb{N}$, and ${\bf y}_r\to {\bf y}$. Using $\rho\left(\boldsymbol{\mathsf{M}}({\bf y}_r,{\bf y})\right)\geq 1$ and the continuity of the spectral radius we deduce that $\rho\left(\boldsymbol{\mathsf{M}}({\bf y},{\bf y})\right)\geq 1$.

$(iii)$ By using Corollary 8.1.30 from \cite{horn} and the fact that $\overline{\bf y}-{\bf y}$ is a positive vector we deduce that $\rho\left(\boldsymbol{\mathsf{M}}(\overline{\bf y},{\bf y})\right)= 1$.

$(iv)$ $\overline{\bf y}-{\bf y}$ is a nonnegative eigenvector of the nonnegative irreducible matrix $\boldsymbol{\mathsf{M}}(\overline{\bf y},{\bf y})$. Consequently, $\overline{\bf y}-{\bf y}$ is a positive eigenvector (see 8.4.P15 from \cite{horn}). We apply $(iii)$.
\end{proof}

In what follows we study the situation in which the matrix function $\boldsymbol{\mathsf{M}}$ is strictly antitone and irreducible.

\begin{thm}\label{irreducible-particular-M-76}
Suppose that the matrix function $\boldsymbol{\mathsf{M}}$ is strictly antitone\footnote{$\boldsymbol{\mathsf{M}}$ is strictly antitone if $({\bf y}, \overline{\bf y})\lneq ({\bf z}, \overline{\bf z}) \Rightarrow M({\bf z}, \overline{\bf z})\lneq M({\bf y}, \overline{\bf y})$.} with respect to $\leq$ and that for all $\overline{\bf y},{\bf y}\in \left(\R_+^*\right)^n$ the matrix $\boldsymbol{\mathsf{M}}(\overline{\bf y},{\bf y})$ is irreducible. 
\begin{enumerate}[(i)]
\item There is no chain with three elements in $\Phi_T$.
\item If $T$ is bounded from above and $\overline{\bf y}, {\bf y}\in \Phi_T$ such that $\overline{\bf y}$, ${\bf y}$, and ${\bf y}_T^{\Box}$ are different two by two, then $\overline{\bf y}$ and ${\bf y}$ are not comparable (with respect to $\leq$). 
\item Suppose that $T$ is bounded from above and $\Phi_T$ has at least two distinct elements.
\begin{enumerate}
\item  $\rho\left(\boldsymbol{\mathsf{M}}({\bf y}_{T}^{\Box},{\bf y}_{T}^{\Box})\right)< 1$ and for ${\bf y}\in \Phi_{T}\backslash\{{\bf y}_{T}^{\Box}\}$ we have $\rho\left(\boldsymbol{\mathsf{M}}({\bf y},{\bf y})\right)> 1$ .
\item If $T$ is a concave function and\footnote{$\rfloor{\bf y},\infty\lceil=\{\overline{\bf y}|{\bf y}< \overline{\bf y}\}$.} ${\bf y}_0\in\bigcup_{{\bf y}\in \Phi_T\backslash\{{\bf y}_T^{\Box}\}}\rfloor{\bf y},\infty\lceil$, then $\emph{Fit}_{T,{\bf y}_0}$ is convergent and its limit is $\omega_T({\bf y}_0)={\bf y}_T^{\Box}$. 
\end{enumerate}
\end{enumerate} 
\end{thm}

\begin{proof}
$(i)$ Suppose that ${\bf y}\lneq \overline{\bf y}\lneq \overline{\overline{\bf y}}$ is a chain of fixed points of $T$. We have $({\bf y},\overline{\bf y})\lneq ({\bf y},\overline{\overline{\bf y}})\lneq (\overline{\bf y},\overline{\overline{\bf y}})$ 
and $\boldsymbol{\mathsf{M}}(\overline{\bf y},\overline{\overline{\bf y}})\lneq \boldsymbol{\mathsf{M}}({\bf y},\overline{\overline{\bf y}})\lneq \boldsymbol{\mathsf{M}}({\bf y},\overline{\bf y})$. These matrices are nonnegative and irreducible. Consequently, $\rho(\boldsymbol{\mathsf{M}}(\overline{\bf y},\overline{\overline{\bf y}}))<\rho(\boldsymbol{\mathsf{M}}({\bf y},\overline{\overline{\bf y}}))<\rho(\boldsymbol{\mathsf{M}}({\bf y},\overline{\bf y}))$, see \cite{horn}, 8.4.P15. From Theorem \ref{q-ineq-sectral-7650-v2} we obtain a contradiction.

$(ii)$ From Theorem \ref{x-Bix-bounded-112} we have that $\overline{\bf y}\lneq {\bf y}_T^{\Box}$ and ${\bf y}\lneq {\bf y}_T^{\Box}$. From $(i)$ we have that $\{\overline{\bf y}, {\bf y}, {\bf y}_T^{\Box}\}$ is not a chain and we deduce that $\overline{\bf y}$ and ${\bf y}$ are not comparable.

$(iii)-(a)$ Let be ${\bf y}_{T}^{\Box},{\bf y}\in \Phi_{T}$ with ${\bf y}\lneq {\bf y}_{T}^{\Box}$. We have $\boldsymbol{\mathsf{M}}({\bf y}_{T}^{\Box},{\bf y}_{T}^{\Box})\lneq \boldsymbol{\mathsf{M}}({\bf y}_T^{\Box},{\bf y})\lneq \boldsymbol{\mathsf{M}}({\bf y},{\bf y})$. Because these matrices are nonnegative and irreducible, using Theorem \ref{q-ineq-sectral-7650-v2},  we obtain that  $\rho\left(\boldsymbol{\mathsf{M}}({\bf y}_{T}^{\Box},{\bf y}_{T}^{\Box})\right)< 1<\rho\left(\boldsymbol{\mathsf{M}}({\bf y},{\bf y})\right)$.

For $(iii)-(b)$ we use $(i)$ and we apply Theorem \ref{domain-of-attraction}.
\end{proof}

\section{The isotone electric systems}\label{increasing-power-563}

In this section we study the fixed points of the function $T_{{\bf k},M}:\left(\R_+^*\right)^n\to \R^n$ given by $T_{{\bf k},M}({\bf y})={\bf k}-M \frac{1}{\bf y}$, $M\in \mathcal{M}_n(\R)$ is a nonnegative matrix and ${\bf k}\in \R^n$. A fixed point of $T_{{\bf k},M}$ is a positive solution of the isotone electric system \eqref{ecuatia-fix-increasing-321}.
The function $T_{{\bf k},M}$ has the following properties that are easy to notice.
\begin{lem}\label{bounded-TkM}
\begin{enumerate}[(i)]
\item $T_{{\bf k},M}$ is bounded from above by ${\bf k}$.
 \item If $P$ is a permutation matrix\footnote{A permutation matrix is a square matrix that has exactly one entry of 1 in each row and each column and 0's elsewhere.} and ${\bf y}\in \left(\R_+^*\right)^{n}$, then $T_{P^T{\bf k},P^TMP}(P^T{\bf y})=P^T T_{{\bf k},M}({\bf y})$.
\end{enumerate}
\end{lem}

For ${\bf y}, \overline{\bf y}\in \left(\R_+^*\right)^{n}$ we have the equality
\begin{equation}\label{M-diag-egalitate-11}
T_{{\bf k},M}(\overline{\bf y})-T_{{\bf k},M}({\bf y})=M\text{diag}\frac{1}{{\bf y}\circ\overline{\bf y}}(\overline{\bf y}-{\bf y}).
\end{equation}
We introduce $\boldsymbol{\mathsf{M}}(\cdot,\cdot):\left(\R_+^*\right)^{n}\times \left(\R_+^*\right)^{n}\to \mathcal{M}_n(\R_+)$ given by 
$\boldsymbol{\mathsf{M}}(\overline{\bf y},{\bf y})=M\text{diag}\frac{1}{{\bf y}\circ\overline{\bf y}}$.

\begin{lem}\label{props-M-09} The above matrix function has the following properties:
\begin{enumerate}[(i)]
\item $\boldsymbol{\mathsf{M}}({\bf y},\overline{\bf y})=\boldsymbol{\mathsf{M}}(\overline{\bf y},{\bf y})$, $\forall {\bf y},\overline{\bf y}\in \left(\R_+^*\right)^{n}$.
\item The matrix function $\boldsymbol{\mathsf{M}}(\cdot,\cdot)$ is antitone with respect to $\leq$.
\item If $M$ is reducible, then $\forall {\bf y},\overline{\bf y}\in \left(\R_+^*\right)^{n}$ the matrix $\boldsymbol{\mathsf{M}}(\overline{\bf y},{\bf y})$ is reducible.
\item If $M$ is irreducible, then $\forall {\bf y},\overline{\bf y}\in \left(\R_+^*\right)^{n}$ the matrix $\boldsymbol{\mathsf{M}}(\overline{\bf y},{\bf y})$ is irreducible and the matrix function $\boldsymbol{\mathsf{M}}(\cdot,\cdot)$ is strictly antitone.
\item If $M$ has on each row at least a positive element, then $\forall {\bf y},\overline{\bf y}\in \left(\R_+^*\right)^{n}$ the matrix $\boldsymbol{\mathsf{M}}(\overline{\bf y},{\bf y})$ has on each row at least a positive element.
\end{enumerate}
\end{lem}

\begin{proof}
For $(iv)$ see Lemma \ref{irreducible-produs-99}.
\end{proof}

\begin{lem}\label{*-increasing-9876} The function which defines the isotone electric system has the properties:
\begin{enumerate}[(i)]
\item $T_{{\bf k},M}$ is a isotone function with respect to $\leq$. If $M$ has on each row a positive element, then $T_{{\bf k},M}$ is a strongly isotone function.
\item $T_{{\bf k},M}$ is a concave function with respect to $\leq$.
\end{enumerate}

\end{lem}

\begin{proof}
$(i)$ If ${\bf y}\leq \overline{\bf y}$, then $\boldsymbol{\mathsf{M}}({\bf y},\overline{\bf y})$ is a nonnegative and $\overline{\bf y}-{\bf y}$ is a nonnegative vector.
For ${\bf y}< \overline{\bf y}$, from Lemma \ref{props-M-09}, then $\boldsymbol{\mathsf{M}}({\bf y},\overline{\bf y})(\overline{\bf y}-{\bf y})$ is a positive vector. 

$(ii)$ For $\lambda\in (0,1)$ we can write
{$$T_{{\bf k},M}(\lambda{\bf y}+(1-\lambda)\overline{\bf y})  - \lambda T_{{\bf k},M}({\bf y})-(1-\lambda)T_{{\bf k},M}(\overline{\bf y}) 
 =\lambda(1-\lambda)M\frac{({\bf y}-\overline{\bf y})\circ ({\bf y}-\overline{\bf y}) }{{\bf y}\circ\overline{\bf y}\circ(\lambda{\bf y}+(1-\lambda)\overline{\bf y})}.$$}
We observe that $\lambda(1-\lambda)>0$ and $\frac{({\bf y}-\overline{\bf y})\circ ({\bf y}-\overline{\bf y}) }{{\bf y}\circ\overline{\bf y}\circ(\lambda{\bf y}+(1-\lambda)\overline{\bf y})}$ is a nonnegative vector.
\end{proof}

In order to delimit  the set $\Phi_{T_{{\bf k},M}}$ we introduce  the vector $\boldsymbol{\Delta}\in \R^n$ with the components $\Delta_i=k_i^2-4M_{ii},\,\,\,i\in \{1,\dots,n\}$. When ${\bf 0}\leq \boldsymbol{\Delta}$ we use the vector $\sqrt{\boldsymbol{\Delta}}$, with the components $\sqrt{\Delta}_i$, to introduce the vectors 
$${\bf y}^{\min}:=\frac{1}{2}({\bf k}-\sqrt{\boldsymbol{\Delta}}),\,\,{\bf y}^{\max}:=\frac{1}{2}({\bf k}+\sqrt{\boldsymbol{\Delta}}).$$

We present a necessary condition for the existence of the fixed points of $T_{{\bf k},M}$. Also, we give a closed ordered interval containing all the fixed points.

\begin{prop}\label{incadrare-fix-crescator-89}
\begin{enumerate}[(i)]
\item A necessary condition for the existence of a fixed point of $\Phi_{T_{{\bf k},M}}$ is ${\bf 0}< {\bf k}$ and ${\bf 0}\leq \boldsymbol{\Delta}$.
\item If the above condition is satisfied, then ${\bf 0}\leq {\bf y}^{\min}$, ${\bf 0}< {\bf y}^{\max}$, and we have the inclusions $\Phi_{T_{{\bf k},M}}\subset \lfloor{\bf y}^{\min},{\bf y}^{\max}\rceil\subset \lfloor{\bf 0},{\bf k}\rceil$.
\end{enumerate}
\end{prop}

\begin{proof} Let be ${\bf y}\in \Phi_{T_{{\bf k},M}}$. The component $i$ verifies
$0<y_i=k_i-\sum\limits_{j=1}^nM_{ij}\frac{1}{y_j}\leq k_i-\frac{M_{ii}}{y_i}\leq k_i.$
We obtain $k_i>0$ and $\Delta_i\geq 0$.
The previous inequalities imply $y_i^{\min}\leq y_i\leq y_i^{\max}$.
\end{proof}

\begin{rem}\label{diagonal-matrix-11} For ${\bf 0}< {\bf k}$ the following statements are equivalent:
\begin{enumerate}[(i)]
\item $M$ is a diagonal matrix with all diagonal entries being positive and $\boldsymbol{\Delta}\in \left(\R_+^*\right)^{n}$.
\item ${\bf y}^{\min},{\bf y}^{\max}\in \Phi_{T_{{\bf k},M}}$.
\end{enumerate}
\end{rem}

By using the values of ${\bf y}^{\max}$, ${\bf y}^{\min}$ and the fixed point iteration sequences we can give necessary and sufficient conditions for the existence of fixed points.

\begin{thm}\label{existenta-k=n-909} Suppose that $M$ is a matrix with all diagonal entries being positive, ${\bf 0}< {\bf k}$, and ${\bf 0}\leq \boldsymbol{\Delta}$. The following statements are equivalent:
\begin{enumerate}[(i)]
\item $\Phi_{T_{{\bf k},M}}\neq  \emptyset$.
\item For all ${\bf y}_0$ with ${\bf y}^{\max}\leq {\bf y}_0$ we have $T(\emph{Fit}_{T_{{\bf k},M},{\bf y}_0})\subset \lfloor{\bf y}^{\min}, \infty\lceil$.
\end{enumerate}
\end{thm} 

\begin{proof} We observe that ${\bf 0}<{\bf y}^{\min}$ and
we prove that $\lfloor {\bf y}^{\max},\infty \lceil \subset S^-_{T_{{\bf k},M}}$. For ${\bf y}^{\max}\leq {\bf y}$ we have $(T_{{\bf k},M}({\bf y}))_i=k_i-\sum\limits_{j=1}^nM_{ij}\frac{1}{y_{j}}\leq k_i-M_{ii}\frac{1}{y_i}\leq y_i.$
From Proposition \ref{incadrare-fix-crescator-89} we have $\Phi_{T_{{\bf k},M}}\subset \lfloor{\bf y}^{\min},{\bf y}^{\max}\rceil$. We apply Theorem \ref{existence-231}.
\end{proof}

From Theorem \ref{x-Bix-bounded-112} we obtain the following result. 

\begin{thm}\label{increasing-dominates-22}
If $\Phi_{T_{{\bf k},M}}\neq \emptyset$, then ${\bf y}_{T_{{\bf k},M}}^{\Box}=(\sup\limits_{{\bf y}\in \Phi_{T_{{\bf k},M}}} y_1, \dots, \sup\limits_{{\bf y}\in \Phi_{T_{{\bf k},M}}} y_n)^t\in \Phi_{T_{{\bf k},M}}$. The dominant fixed point ${\bf y}_{T_{{\bf k},M}}^{\Box}$ dominates all the fixed points of $T_{{\bf k},M}$. Also, it dominates the $\omega$-limit set of $T$ ($\Omega_{T_{{\bf k},M}}\subset \lfloor {\bf 0},{\bf y}_{T_{{\bf k},M}}^{\Box}\rceil$). For ${\bf y}_0$, with ${\bf y}_{T_{{\bf k},M}}^{\Box}\leq {\bf y}_0$, the sequence $\emph {Fit}_{T_{{\bf k},M},{\bf y}_0}$ is convergent and its limit is $\omega_{T_{{\bf k},M}}({\bf y}_0)={\bf y}_{T_{{\bf k},M}}^{\Box}$.

\end{thm}

\begin{rem}\label{permutation-102}
For $P$ a permutation matrix, we have $P^T{\bf y}_{T_{{\bf k},M}}^{\Box}={\bf y}_{T_{P^T{\bf k},P^TMP }}^{\Box}$.
\end{rem}

\begin{thm}\label{raza-mai-mica-1}
If $\Phi_{T_{{\bf k},M}}\neq  \emptyset$, then $\rho\left(M\emph{diag}\frac{1}{{\bf y}_{T_{{\bf k},M}}^{\Box}\circ{\bf y}_{T_{{\bf k},M}}^{\Box}}\right)\leq 1$.
\end{thm}

\begin{proof}
When $M$ has on each row a positive element we apply Theorem \ref{inegalitate-yBox-10}, Lemma \ref{bounded-TkM}, and Lemma \ref{props-M-09}.

Next, we consider the case when some rows have all elements equal to zero. There exists $P$ a permutation matrix such that $P^TMP=\begin{pmatrix} O_{s\times s} & O_{s\times (n-s)} \\ A & B \end{pmatrix}$, $A\in \mathcal{M}_{(n-s)\times s}(\R_+)$, $B\in \mathcal{M}_{(n-s)\times (n-s)}(\R_+)$, and on each row of the matrix $\begin{pmatrix} A & B\end{pmatrix}$ we have a positive element. The vector ${\bf y}=({\bf y}_1, {\bf y}_2)^t$, ${\bf y}_1\in \R^s$, ${\bf y}_2\in \R^{n-s}$ is a fixed point of $T_{P^T{\bf k},P^TMP }$ if and only if ${\bf y}_1={\bf k}_1$ and ${\bf y}_2\in \Phi_{T_{{\bf k}_2-A\frac{1}{{\bf k}_1},B}}$ with $P^T{\bf k}=({\bf k}_1, {\bf k}_2)^t$. We deduce that ${\bf y}_{T_{P^T{\bf k},P^TMP }}^{\Box}=({\bf k}_1, {\bf y}_{_{{\bf k}_2-A\frac{1}{{\bf k}_1},B}}^{\Box})^t$. From Remark \ref{permutation-102} we have\footnote{$P^T{\bf y}\circ P^T\overline{\bf y}=P^T({\bf y}\circ \overline{\bf y})$, $P^T\frac{1}{\bf y}=\frac{1}{P^T {\bf y}}$}
$\frac{1}{{\bf y}_{T_{P^T{\bf k},P^TMP }}^{\Box}\circ {\bf y}_{T_{P^T{\bf k},P^TMP }}^{\Box}}=\frac{1}{P^T{\bf y}_{T_{{\bf k},M}}^{\Box}\circ P^T{\bf y}_{T_{{\bf k},M}}^{\Box}}=P^T\frac{1}{{\bf y}_{T_{{\bf k},M}}^{\Box}\circ {\bf y}_{T_{{\bf k},M}}^{\Box}}$
and\footnote{$\text{diag}(P^T{\bf y})=P^T\text{diag}({\bf y})P$.} 
$P^TMP\text{diag}\frac{1}{{\bf y}_{T_{P^T{\bf k},P^TMP }}^{\Box}\circ {\bf y}_{T_{P^T{\bf k},P^TMP }}^{\Box}}=P^TM\text{diag}\frac{1}{{\bf y}_{T_{{\bf k},M}}^{\Box}\circ {\bf y}_{T_{{\bf k},M}}^{\Box}}P.$
We deduce that 
$\rho\left(P^TMP\text{diag}\frac{1}{{\bf y}_{T_{P^T{\bf k},P^TMP }}^{\Box}\circ {\bf y}_{T_{P^T{\bf k},P^TMP }}^{\Box}}\right)=\rho\left(M\text{diag}\frac{1}{{\bf y}_{T_{{\bf k},M}}^{\Box}\circ {\bf y}_{T_{{\bf k},M}}^{\Box}}\right).$
We obtain that 
$\rho\left(M\text{diag}\frac{1}{{\bf y}_{T_{{\bf k},M}}^{\Box}\circ {\bf y}_{T_{{\bf k},M}}^{\Box}}\right)=\rho\left(B\text{diag}\frac{1}{{\bf y}_{_{{\bf k}_2-A\frac{1}{{\bf k}_1},B}}^{\Box}\circ {\bf y}_{_{{\bf k}_2-A\frac{1}{{\bf k}_1},B}}^{\Box}}\right).$
In the case when $B=O_{(n-s)\times (n-s)}$ we obtain that the spectral radius is $0<1$.

In the case when $B$ has on each row a positive element we apply the above result. 

In the case when $B$ has a row with all elements equal to zero we repeat the above reduction. After a finite number of steps we obtain a matrix ''$B$'' which is the zero matrix or it is a matrix with a positive element on each row. 
\end{proof}

\begin{rem}
We observe that the Jacobian matrix of $T_{{\bf k},M}$ is 
$J_{T_{{\bf k},M}}({\bf y})=M\text{diag}\frac{1}{{\bf y}\circ{\bf y}}$.
When $\Phi_{T_{{\bf k},M}}\neq  \emptyset$ and $\rho\left(M\text{diag}\frac{1}{{\bf y}_{T_{{\bf k},M}}^{\Box}\circ{\bf y}_{T_{{\bf k},M}}^{\Box}}\right)< 1$ we have that the matrix $J_{T_{{\bf k},M}}({\bf y}_{T_{{\bf k},M}}^{\Box})$ has all the eigenvalues with the modulus $<1$. From the theory of discrete dynamical systems we have that ${\bf y}_{T_{{\bf k},M}}^{\Box}$ is asymptotically stable (see \cite{holmes}) for the fixed point iteration method. 
\end{rem}

Using Theorem \ref{q-ineq-sectral-7650-v2}, Theorem \ref{x-Bix-bounded-112}, and the fact that $\boldsymbol{\mathsf{M}}$ is antitone,  when $T_{{\bf k},M}$ have at least two fixed points, we have the following inequalities involving spectral radius. 

\begin{thm}\label{spectral-radius-q=n-879}
Let be ${\bf y},\overline{\bf y}\in \Phi_{T_{{\bf k},M}}$ such that ${\bf y}\neq \overline{\bf y}$. The following results hold.
\begin{enumerate}[(i)]
\item $\rho\left(M\emph{diag}\frac{1}{{\bf y}\circ\overline{\bf y}}\right)\geq 1$.
\item If ${\bf y}< \overline{\bf y}$, then $\rho\left(M\emph{diag}\frac{1}{{\bf y}\circ\overline{\bf y}}\right)= 1$.
\item If ${\bf y}\neq {\bf y}_{T_{{\bf k},M}}^{\Box}$, then $\rho\left(M\emph{diag}\frac{1}{{\bf y}\circ{\bf y}}\right)\geq 1$.
\end{enumerate}
\end{thm}

In what follows we study the case when the matrix $M$ is irreducible.

\begin{thm}\label{q=n-irreducible-045}
Suppose that $n>1$, $M$ is an irreducible matrix and $\Phi_{T_{{\bf k},M}}\neq \emptyset$.
\begin{enumerate}[(i)]
\item If ${\bf y}\in \Phi_{T_{{\bf k},M}}$, ${\bf y}\neq {\bf y}_{T_{{\bf k},M}}^{\Box}$, then ${\bf y}< {\bf y}_{T_{{\bf k},M}}^{\Box}$ and $\rho\left(M\emph{diag}\frac{1}{{\bf y}\circ{\bf y}_{T_{{\bf k},M}}^{\Box}}\right)= 1$.
\item If ${\bf y},\overline{\bf y}\in \Phi_{T_{{\bf k},M}}$ such that ${\bf y}$, $\overline{\bf y}$, and ${\bf y}_{T_{{\bf k},M}}^{\Box}$ are different two by two, then ${\bf y}$ and $\overline{\bf y}$ are not comparable. 
\item If $\Phi_{T_{{\bf k},M}}$ has at least two distinct elements, then  $\rho\left(M\emph{diag}\frac{1}{{\bf y}_{T_{{\bf k},M}}^{\Box}\circ{\bf y}_{T_{{\bf k},M}}^{\Box}}\right)< 1$ and  $\rho\left(M\emph{diag}\frac{1}{{\bf y}\circ{\bf y}}\right)> 1$ for all ${\bf y}\in \Phi_{T_{{\bf k},M}}\backslash\{{\bf y}_{T_{{\bf k},M}}^{\Box}\}$.
\item If $\Phi_{T_{{\bf k},M}}$ has at least two elements and ${\bf y}_0\in\bigcup_{{\bf y}\in \Phi_{T_{{\bf k},M}}\backslash\{{\bf y}_{T_{{\bf k},M}}^{\Box}\}}\rfloor{\bf y},\infty\lceil$, then $\emph{Fit}_{{T_{{\bf k},M}},{\bf y}_0}$ is convergent and its limit is $\omega_{T_{{\bf k},M}}({\bf y}_0)={\bf y}_{T_{{\bf k},M}}^{\Box}$.
\item ${\bf y}_{T_{{\bf k},M}}^{\Box}$ is an isolated fixed point of $T_{{\bf k},M}$.
\end{enumerate}
\end{thm}

\begin{proof}
For $(i)$ we apply Theorem \ref{x-Bix-bounded-112} and Theorem \ref{q-ineq-sectral-7650-v2}.
For $(ii)$ and $(iii)$ we use Lemma \ref{props-M-09} and Theorem \ref{irreducible-particular-M-76}.
The statement $(iv)$ is the consequence of Lemma \ref{props-M-09}, Lemma \ref{*-increasing-9876}, and Theorem \ref{irreducible-particular-M-76}. The statement $(v)$ is obtained by using $(iii)$ and Theorem \ref{q-ineq-sectral-7650-v2}.
\end{proof}

\begin{rem}
When $n>1$, $M$ is irreducible and $\Phi_{T_{{\bf k},M}}$ has at least two elements
we have $\lfloor {\bf y}_{T_{{\bf k},M}}^{\Box},\infty\lceil\subset \bigcup_{{\bf y}\in \Phi_{T_{{\bf k},M}}\backslash\{{\bf y}_{T_{{\bf k},M}}^{\Box}\}}\rfloor{\bf y},\infty\lceil$.
From the theory of discrete dynamical systems we have that a fixed point ${\bf y}$, ${\bf y}\neq {\bf y}_{T_{{\bf k},M}}^{\Box}$  is not stable for the fixed point iteration method. 
\end{rem}

\begin{rem}\label{reducere-ireductibil}
In the case when $M$ is a reducible matrix we can reduce the equation \eqref{ecuatia-fix-increasing-321} (the determination of the fixed points of $T_{{\bf k},M}$) to a system of equations of the form \eqref{ecuatia-fix-increasing-321} and in each equation we have an irreducible matrix (possibly a 1-by-1 zero matrix). From Lemma \ref{bounded-TkM}, the equation \eqref{ecuatia-fix-increasing-321} is equivalent with an equation of the same form in which the matrix $M$ is in irreducible normal form (see Section \ref{irreducible-matrices}). 

If $M$ is in the irreducible normal form \eqref{irreducible-normal-form-01}, then \eqref{ecuatia-fix-increasing-321} becomes
\begin{equation}
\begin{cases}
{\bf k}_1-M_{11}\frac{1}{{\bf y}_1} -M_{12}\frac{1}{{\bf y}_2}-\dots -M_{1\,s-1}\frac{1}{{\bf y}_{s-1}}-M_{1s}\frac{1}{{\bf y}_s}={\bf y}_1 \\
{\bf k}_2-M_{22}\frac{1}{{\bf y}_2}- \dots -M_{2\,s-1}\frac{1}{{\bf y}_{s-1}} -M_{2s}\frac{1}{{\bf y}_s}={\bf y}_2 \\
\dots \\
{\bf k}_{s-1}-M_{s-1\,s-1}\frac{1}{{\bf y}_{s-1}}-M_{s-1\,s}\frac{1}{{\bf y}_{s}}={\bf y}_{s-1} \\
{\bf k}_s-M_{ss}\frac{1}{{\bf y}_s}={\bf y}_s,
\end{cases}
\end{equation}
where ${\bf y}=({\bf y}_1,\dots,{\bf y}_s)^T$, ${\bf y}_i\in \R^{n_i}$. 
To determine $\Phi_{T_{{\bf k},M}}$ we will follow these steps:
\begin{itemize}
\item We find $\Phi_{T_{{\bf k}_s,M_{ss}}}$.
\item For ${\bf y}_s\in \Phi_{T_{{\bf k}_s,M_{ss}}}$ we find $\Phi_{T_{\widetilde{\bf k}_{s-1},M_{s-1\,s-1}}}$, $\widetilde{\bf k}_{s-1}={\bf k}_{s-1}-M_{s-1\,s}\frac{1}{{\bf y}_s}$.
\item ...
\item For ${\bf y}_s\in \Phi_{T_{{\bf k}_s,M_{ss}}}$, ..., ${\bf y}_2\in \Phi_{T_{\widetilde{\bf k}_2,M_{22}}}$, $\widetilde{\bf k}_{2}={\bf k}_{2}-M_{23}\frac{1}{{\bf y}_3}-\dots - M_{2s}\frac{1}{{\bf y}_s}$   we find $ \Phi_{T_{\widetilde{\bf k}_1,M_{11}}}$, $\widetilde{\bf k}_{1}={\bf k}_{1}-M_{12}\frac{1}{{\bf y}_2}-\dots - M_{1s}\frac{1}{{\bf y}_s}$.
\end{itemize}
We observe that ${\bf y}_{T_{{\bf k},M}}^{\Box}$ has the components ${\bf y}^{\Box}_{T_{{\bf k}^{\Box}_1,M_{11}}}$, ${\bf y}^{\Box}_{T_{{\bf k}^{\Box}_2,M_{22}}}$, ..., ${\bf y}^{\Box}_{T_{{\bf k}^{\Box}_{s-1},M_{s-1\,s-1}}}$, ${\bf y}^{\Box}_{T_{{\bf k}_s,M_{ss}}}$, where ${\bf k}^{\Box}_{1}={\bf k}_{1}-M_{12}\frac{1}{{\bf y}^{\Box}_{T_{{\bf k}^{\Box}_2,M_{22}}}}-\dots - M_{1s}\frac{1}{{\bf y}^{\Box}_{T_{{\bf k}_s,M_{ss}}}}$, ..., ${\bf k}_{s-1}^{\Box}={\bf k}_{s-1}-M_{s-1\,s}\frac{1}{{\bf y}^{\Box}_{T_{{\bf k}_s,M_{ss}}}}$.
\end{rem}

\begin{rem}
If $M$ is a nonnegative, invertible matrix, $M^{-1}$ is nonnegative, and ${\bf z}=\frac{1}{\bf y}$, then the equation \eqref{ecuatia-fix-increasing-321} is equivalent with the equation 
\begin{equation}
{\bf z}=M^{-1}{\bf k}-M^{-1}\frac{1}{\bf z},
\end{equation} 
which has the same form as the equation \eqref{ecuatia-fix-increasing-321}. 
It has been proved that a nonnegative matrix has a nonnegative inverse if and only if its entries are all zero except for a single positive entry in each row and column (see \cite{brown}, \cite{demarr}). The matrix $M$ is
the product of a permutation matrix by a diagonal matrix. In this case, in addition to the dominant fixed point we have a fixed point dominated by all other fixed points (see also Remark \ref{diagonal-matrix-11}) . 
\end{rem}

\subsection{The case $n=1$}

In this case we have $T_{k,M}:\R_+^*\to \R$, $T_{k,M}(y)=k-\frac{M}{y}$ with $k\in \R$ and $M\in \R_+$.

For $M=0$ the function $T_{k,M}$ has fixed points if and only if $k>0$. If the above inequality is satisfied, then $\Phi_{T_{k,0}}=\{k\}$.

For $M>0$ we have $\Phi_{T_{k,M}}\neq \emptyset$ if and only if $k\geq 2\sqrt{M}$. If the above inequality is satisfied, then $\Phi_{T_{k,M}}=\{y^{\Box},y^*\}$ with $y^{\Box}=\frac{1}{2}(k+\sqrt{k^2-4M})$ and $y^{*}=\frac{1}{2}(k-\sqrt{k^2-4M})$. 

Figure \ref{1D-dynamics} shows the dynamics generated by of the fixed point iteration method.

\begin{figure}[h]
\begin{center}
\includegraphics[width=15cm]{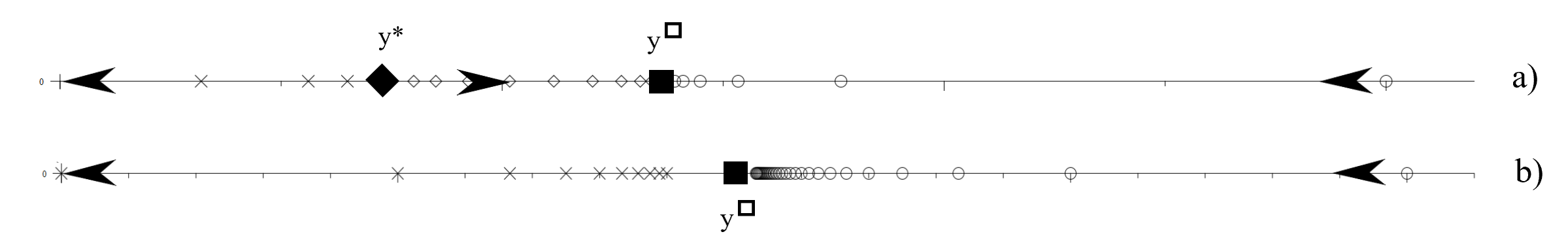}
\caption{The 1-D fixed point iteration sequences.
a) $k> 2\sqrt{M}$; b) $k=2\sqrt{M}$.}\label{1D-dynamics}
\end{center}
\end{figure}

\subsection{Numerical simulation for the steady states of a DC linear circuit with two CPLs}

The mathematical representation of the constant steady states for a DC linear circuit with two CPLs, which is first studied in \cite{polyak} and then re-studied in \cite{matveev}, is
\begin{equation}\label{ecuatie-Matveev-90}
\begin{pmatrix} v_1 \\ v_2 \end{pmatrix}= \begin{pmatrix} E \\ E \end{pmatrix}-\begin{pmatrix} r_1 P_1 & r_1 P_2 \\ r_1 P_1 & (r_1+r_2)P_2 \end{pmatrix}\begin{pmatrix} \frac{1}{v_1} \\ \frac{1}{v_2} \end{pmatrix}.
\end{equation}
In the above system $v_1, v_2$ are the voltage of the capacitors, $P_1, P_2$ are the power of the CPLs, $r_1$ and $r_2$ are the line resistances, and $E$ is the voltage source. We work with the following numerical values taken from \cite{matveev}: 
$E=24\,V$, $r_1=0.04\,\Omega$, and $r_2=0.06\,\Omega$.\medskip

{\bf I. The case $P_1=500\,W$, $P_2=450\,W$.}

In this case we have $\Delta=\begin{pmatrix} 496 \\ 396 \end{pmatrix}$, ${\bf v}^{\min}=\begin{pmatrix}  0.86 \\ 2.05 \end{pmatrix}$, and ${\bf v}^{\max}=\begin{pmatrix}  23.13 \\ 21.94 \end{pmatrix}$. The necessary conditions for the existence of a steady state presented in Proposition \ref{incadrare-fix-crescator-89} are satisfied. All these points are contained in $\lfloor {\bf v}^{\min},{\bf v}^{\max}\rceil$(Proposition \ref{incadrare-fix-crescator-89}). The fixed point iteration sequence which starts from ${\bf v}^{\max}$ is convergent and its limit is the dominant fixed point ${\bf v}^{\Box}=\omega({\bf v}^{\max})=\begin{pmatrix}  22.94 \\ 20.95 \end{pmatrix}$. This point is a solution of \eqref{ecuatie-Matveev-90} and it dominates all other positive solutions (Theorem \ref{increasing-dominates-22}). By computation we obtain that this system has two positive solutions: ${\bf v}^{\Box}$ and ${\bf v}^{*}=\begin{pmatrix}  14.45 \\ 2.20 \end{pmatrix}$. The fixed point iteration sequence which start from $\rfloor {\bf v}^*, \infty\lceil$ is convergent and its limit is ${\bf v}^{\Box}$ (Theorem \ref{q=n-irreducible-045}).

{\bf II. The case $P_1=3000\,W$, $P_2=1000\,W$.}
In this case we have $\Delta=\begin{pmatrix} 96 \\ 196 \end{pmatrix}$, ${\bf v}^{\min}=\begin{pmatrix}  7.10 \\ 5.36 \end{pmatrix}$, and ${\bf v}^{\max}=\begin{pmatrix}  16.89 \\ 18.63 \end{pmatrix}$. The necessary conditions for the existence of a steady state presented in Proposition \ref{incadrare-fix-crescator-89} are satisfied. When we analyze the fixed point iteration sequence starting from ${\bf v}^{\max}$ we observe that the term ${\bf v}_3=
\begin{pmatrix} 8.76 \\ 0.42 \end{pmatrix}$ is not comparable with ${\bf v}^{\min}$. We deduce, from Theorem \ref{existenta-k=n-909}, that the system \eqref{ecuatie-Matveev-90} has no positive solutions.
\begin{figure}[h]
\includegraphics[width=10cm]{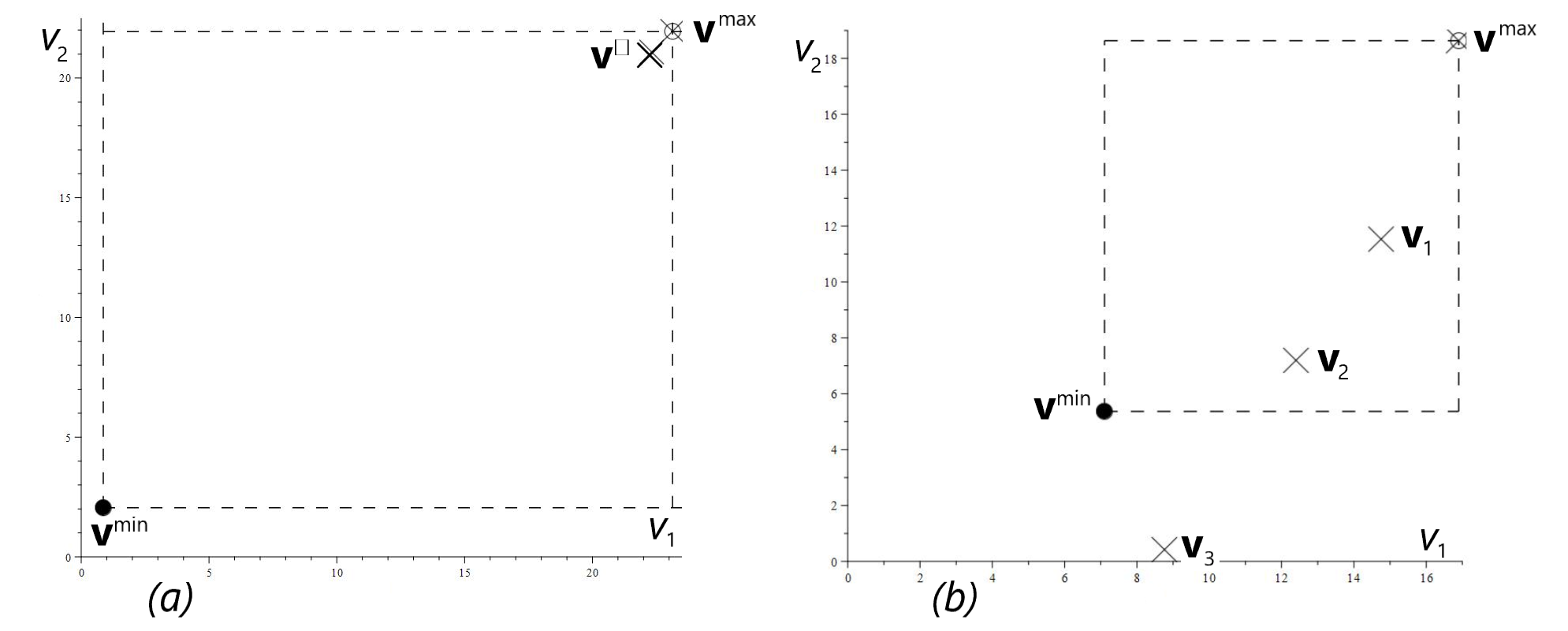}
\centering
\caption{{\small The fixed point iteration sequence which start from ${\bf v}^{\max}$. $(a)$ - $P_1=500\,W$, $P_2=450\,W$; $(b)$ - $P_1=3000\,W$, $P_2=1000\,W$.}}\label{2-CPLs-figure-99}
\centering
\end{figure}

Figure \ref{2-CPLs-figure-99} present the fixed point iteration sequences which start from ${\bf v}^{\max}$.

\subsection{Conclusions}

We present a practical method, using a fixed point iteration sequence, to decide whether an isotone electric systems has a steady state (Theorem \ref{existenta-k=n-909}). When their existence is assured
a dominant steady state may be highlighted. The dominant steady state is the limit of some sequences generated by the fixed point iteration method (Theorem \ref{increasing-dominates-22}) and we specify a part of its domain of attraction (Theorem \ref{increasing-dominates-22} and Theorem \ref{q=n-irreducible-045}).

We pay special attention to the case for which the function is defined by an irreducible matrix. In this case we have:

\noindent - the dominant steady state is an isolated steady state (Theorem \ref{q=n-irreducible-045});

\noindent - a steady state, other than the dominant steady state, is strictly smaller than the dominant steady state (Theorem \ref{q=n-irreducible-045});

\noindent - two steady states, different from the dominant dominant steady state, are incomparable (Theorem \ref{q=n-irreducible-045});

\noindent - a fixed point iteration sequence which start from a vector strictly greater than a steady state is convergent and its limit is the dominant steady state (Theorem \ref{q=n-irreducible-045}). 

When the matrix which appears in the isotone electric system is reducible, the determination of the steady states can be reduced to the determination of the steady states of some isotone electric systems defined by irreducible matrices (Remark \ref{reducere-ireductibil}).

\appendix

\section{Irreducible matrices}\label{irreducible-matrices}

A matrix $A\in \mathcal{M}_{n}(\R)$, $n>1$, is reducible (see Definition 6.2.21, \cite{horn}) if there is a permutation matrix $P\in \mathcal{M}_{n}(\R)$ such that 
\begin{equation}
P^TAP=\begin{pmatrix}
B & C \\
O_{(n-r)\times r} & D
\end{pmatrix}\,\,
\text{and}\,\,1\leq r\leq n-1.
\end{equation}
A lower-left $(n-r)\times r$ block of zero entries can be created by some sequence of row and column interchanges. 
A matrix $A\in \mathcal{M}_{n}(\R)$, $n>1$, is irreducible if it is not reducible (see Definition 6.2.22, \cite{horn}). All matrices of $\mathcal{M}_1(\R)$ are irreducible.

\begin{lem}\label{irreducible-produs-99}
Let be $A\in \mathcal{M}_{n}(\R)$ an irreducible matrix.
\begin{enumerate}[(i)]
\item If ${\bf d}\in \left(\R_+^*\right)^n$, then $A\,(\emph{diag}\,{\bf d})$ is irreducible.
\item If $A$ is nonnegative, $n>1$, ${\bf d}, {\bf e}\in \left(\R_+^*\right)^n$ with ${\bf d}\lneq {\bf e}$, then $A\,(\emph{diag}\,{\bf d})\lneq A\,(\emph{diag}\,{\bf e})$.
\end{enumerate}
\end{lem}

\begin{proof}
$(i)$ If $B:=A\,(\text{diag}\,{\bf d})$, then $B_{ij}=d_jA_{ij}$, $i,j\in \{1,\dots,n\}$ and $B_{ij}=0\Leftrightarrow A_{ij}=0$.

$(ii)$ It is easy to observe that  $A\,(\text{diag}\,{\bf d})\leq A\,(\text{diag}\,{\bf e})$.  
There exists $r\in \{1,\dots,n\}$ such that $d_r<e_r$. Because $A$ is irreducible we have an element $A_{ir}>0$, $i\in \{1,\dots,n\}$. Because $A_{ir}d_r<A_{ir}e_r$ we obtain the announced result.
\end{proof}

The matrix $M\in \mathcal{M}_n(\R)$ is in irreducible normal form, see \cite{horn}, if it is block upper triangular, and each diagonal block is irreducible (possibly a 1-by-1 zero matrix); more precisely,
\begin{equation}\label{irreducible-normal-form-01}
M=\begin{pmatrix} M_{11} & M_{12} & \dots & M_{1\,s-1} & M_{1s} \\
O & M_{22} & \dots & M_{1\,s-1} & M_{2s} \\
\vdots & \vdots & \ddots & \vdots & \vdots \\
O & O & \dots & M_{s-1\,s-1} & M_{s-1\,s} \\
O & O & \dots & O & M_{ss}
\end{pmatrix}, \,\,1\leq s\leq n,
\end{equation}
with $M_{11}\in \mathcal{M}_{n_1}(\R_+),\dots,M_{ss}\in \mathcal{M}_{n_s}(\R_+)$ irreducible matrices, $n_1\geq 1$, ..., $n_s\geq 1$, and $\sum_{i=1}^sn_i=n$.

An irreducible normal form of $A\in \mathcal{M}_n(\R)$ is a matrix $B=P^TAP$ in an irreducible normal form, with $P$ a permutation matrix. 

An irreducible normal form of the matrix $A$ is not necessarily unique.
If $A\in \mathcal{M}_n(\R)$ is irreducible then it is in irreducible normal form.
For $A\in \mathcal{M}_n(\R)$ there exists $B\in \mathcal{M}_n(\R)$ such that $B$ is the irreducible normal form of $A$.

\section{The fixed points, the fixed point iteration sequences, and the $\omega$-limit set}\label{fixed-point-iteration-section}

Let $T:\left(\R_+^*\right)^n\to \R^n$ be a continuous function. The vector ${\bf y}\in\left(\R_+^*\right)^n$ is a fixed point of $T$ if $T({\bf y})={\bf y}$. We denote by $\Phi_T$ the set of fixed points of $T$.

We consider the sequence of open sets
$D_{T^{(1)}}=\left(\R_+^*\right)^n,\,\,D_{T^{(r+1)}}=T^{-1}\left(D_{T^{(r)}}\right)\bigcap D_{T^{(r)}}\subset \left(\R_+^*\right)^n ,\,r\geq 1,$
and the composite functions $T^{(r)}:D_{T^{(r)}}\to \R^n$ given by $T^{(r)}=\stackrel{r\,\text{times}}{T\circ\dots\circ T}$. Also, we denote by $D_T^{\infty}=\bigcap_{r\in \mathbb{N}^*}D_{T^{(r)}}$. 

For ${\bf y}\in \left(\R_+^*\right)^n$ we construct the number $N_{\bf y}\in \mathbb{N}\cup\{\infty\}$, 
$$N_{\bf y}=\begin{cases}
\infty & \text{if}\,\,{\bf y}\in D_T^{\infty} \\
\min\{r\in \mathbb{N}\,|\,{\bf y}\notin D_{T^{(r+1)}}\} & \text{if}\,\,{\bf y}\in \left(\R_+^*\right)^n\backslash D_T^{\infty}.
\end{cases}$$

The fixed point iteration sequence $\text{Fit}_{T,{\bf y}_0}:=({\bf y}_r)_{r\in \{0,\dots,N_{{\bf y}_0}\}}$ starts from the vector  ${\bf y}_0\in \left(\R_+^*\right)^n$ and it verifies the iteration
\begin{equation}
{\bf y}_{r+1}=T({\bf y}_r),\,\,\,r<N_{{\bf y}_0}.
\end{equation}

It is easy to observe that we have the following results.
\begin{lem}\label{isotone-fixed-sequence-23}
Suppose that $T$ is isotone. Let be ${\bf y}_0,\overline{\bf y}_0\in \left(\R_+^*\right)^n$ such that ${\bf y}_0\leq \overline{\bf y}_0$.
\begin{enumerate}[(i)]
\item $N_{{\bf y}_0}\leq N_{\overline{\bf y}_0}$. If ${\bf y}_0\in D_{T}^{\infty}$, then $\lfloor{\bf y}_0,\infty\lceil\subset D_{T}^{\infty}$.
\item If $r\in \mathbb{N}$ with $r\leq N_{{\bf y}_0}$, then ${\bf y}_r\leq \overline{\bf y}_r$.
\end{enumerate}
\end{lem}

For ${\bf y}_0\in D_T^{\infty}$ we work with the $\omega$-limit set of $\text{Fit}_{T,{\bf y}_0}$ defined by 
\begin{equation}
\omega_T({\bf y}_0)=\{{\bf y}\in \R_+^n\,|\,\exists\,\text{the subsequence}\, ({\bf y}_{k_q})_{q\in \mathbb{N}}\,\text{of}\, \text{Fit}_{T,{\bf y}_0}\, \text{such that} \,{\bf y}_{k_q}\stackrel{q\to \infty}{\longrightarrow} {\bf y}\}.
\end{equation}
From the continuity of $T$ we obtain that the set $\omega_T({\bf y}_0)$ is invariant under $T$. When $\text{Fit}_{T,{\bf y}_0}$ is bounded from above by ${\bf w}$, then it is contained in the compact set $\lfloor{\bf 0},{\bf w}\rceil$, and, consequently, $\omega_T({\bf y}_0)\neq \emptyset$. 
When $\text{Fit}_{T,{\bf y}_0}$ is convergent we denote by $\omega_T({\bf y}_0)$ its limit. If $\text{Fit}_{T,{\bf y}_0}$ is convergent and $\omega_T({\bf y}_0)\in \left(\R_+^*\right)^n$, then $\omega_T({\bf y}_0)$ is a fixed point of $T$. 
 
The $\omega$-limit set of $T$ (with respect to the fixed point iteration method) is:
\begin{equation}
\Omega_T=\{{\bf y}\in \R_+^n\,|\,\exists {\bf y}_0\in D_T^{\infty}\,\text{such that}\,{\bf y}\in \omega_T({\bf y}_0)\}.
\end{equation}
It is easy to observe that we have $\Phi_T\subset\Omega_T$.


\begin{thebibliography}{99}
\bibitem{polyak} N. Barabanov, R. Ortega, R. Gri\~n\'o, B. Polyak, {\it On Existence and Stability of Equilibria of Linear
Time-Invariant Systems With Constant Power Loads}, IEEE Trans. Circuits Syst. I, Reg. Papers, 63, no.1 (2016),
pp. 114–121.
\bibitem{brown} T.A. Brown, M. Juncosa, V. Klee, {\it Invertibly positive linear operators on
spaces of continuous functions}, Math. Ann., 183 (1969), pp. 105-114.
\bibitem{chandler} D. Chandler, {\it The norm of the Schur product operation}, Numer. Math., 4, no. 1 (1962), pp. 343–44. 
\bibitem{demarr} R. DeMarr, {\it Nonnegative matrices with nonnegative inverses}, Proc. Amer. Math. Soc., 35, no. 1 (1972), pp. 307-308.
\bibitem{ehrgott} M. Ehrgott, {\it Multicriteria Optimization. Second edition, Springer}, 2005.
\bibitem{holmes} J. Guckenheimer, P. Holmes, {\it Nonlinear Oscilations, Dynamical Systems, and Bifurcations of Vector Fields (corrected third printing), Springer-Verlag}, 1990. 
\bibitem{horn} R. A. Horn, C. R. Johnson, {\it Matrix Analysis, Second Edition, Cambridge University Press}, 2013.
\bibitem{matveev} A. S. Matveev, J. E. Machado, R. Ortega, J. Schiffer, A. Pyrkin, {\it A Tool for Analysis of Existence of Equilibria and Voltage Stability in Power Systems With Constant Power Loads}, IEEE Trans. Automat. Contr., 65, no. 11 (2020), pp. 4726-4740.
\bibitem{micchelli} C. A. Micchelli, R. A. Willoughby, {\it On Functions Which Preserve the Class of Stieltjes Matrices}, Linear Algebra Appl., 23 (1979), pp. 141-156. 
\bibitem{sanchez} S. Sanchez, R. Ortega, R. Gri\~n\'o, G. Bergna, M. Molinas-Cabrera, {\it Conditions for Existence of Equilibrium Points of Systems with Constant Power Loads}, IEEE Trans. Circuits Syst. I, Reg. Papers, vol. 61,
no. 7 (2014), pp. 2204–2211.
\end{thebibliography}
\end{document}